\documentclass{article}
\usepackage{citesort}
\usepackage{enumerate}
\newcounter{savenumi}

\newtheorem{theoremfoo}{Theorem}
\newenvironment{theorem}{\pagebreak[1]\begin{theoremfoo}}{\end{theoremfoo}}

\newtheorem{propositionfoo}[theoremfoo]{Proposition}

\newtheorem{lemmafoo}[theoremfoo]{Lemma}
\newenvironment{lemma}{\pagebreak[1]\begin{lemmafoo}}{\end{lemmafoo}}
\newtheorem{conjecturefoo}[theoremfoo]{Conjecture}

\newtheorem{corollaryfoo}[theoremfoo]{Corollary}
\newenvironment{corollary}{\pagebreak[1]\begin{corollaryfoo}}{\end{corollaryfoo}}
\newtheorem{exercisefoo}{Exercise}

\newtheorem{openfoo}[theoremfoo]{Question}

\newtheorem{nttn}[theoremfoo]{Notation}

\newtheorem{dfntn}[theoremfoo]{Definition}
\newenvironment{definition}{\pagebreak[1]\begin{dfntn}\rm}{\end{dfntn}}

\newenvironment{proof}
    {\pagebreak[1]{\narrower\noindent {\bf Proof:\quad\nopagebreak}}}{\QED}

\newcommand{\PH}{{\rm PH}}

\newcommand{\NP}{{\rm NP}}

\renewcommand{\P}{{\rm P}}      

\def\nre.{$n$\/-r.e.}

\newcommand{\NEXP}{{\rm NEXP}}

\newtheorem{factfoo}[theoremfoo]{Fact}

\newcommand{\BPP}{{\rm BPP}}

\newcommand{\squeeze}{
\textwidth 6in
\textheight 8.8in
\oddsidemargin 0.2in
\topmargin -0.4in
}

\newtheorem{propertyfoo}[theoremfoo]{Property}

\makeatletter

\def\@makechapterhead#1{ \vspace*{50pt} { \parindent 0pt \raggedright 
 \ifnum \c@secnumdepth >\m@ne \huge\bf \@chapapp{} \thechapter. \par 
 \vskip 20pt \fi \Huge \bf #1\par 
 \nobreak \vskip 40pt } }

\def\@sect#1#2#3#4#5#6[#7]#8{\ifnum #2>\c@secnumdepth
     \def\@svsec{}\else 
     \refstepcounter{#1}\edef\@svsec{\csname the#1\endcsname.\hskip 1em }\fi
     \@tempskipa #5\relax
      \ifdim \@tempskipa>\z@ 
        \begingroup #6\relax
          \@hangfrom{\hskip #3\relax\@svsec}{\interlinepenalty \@M #8\par}
        \endgroup
       \csname #1mark\endcsname{#7}\addcontentsline
         {toc}{#1}{\ifnum #2>\c@secnumdepth \else
                      \protect\numberline{\csname the#1\endcsname}\fi
                    #7}\else
        \def\@svsechd{#6\hskip #3\@svsec #8\csname #1mark\endcsname
                      {#7}\addcontentsline
                           {toc}{#1}{\ifnum #2>\c@secnumdepth \else
                             \protect\numberline{\csname the#1\endcsname}\fi
                       #7}}\fi
     \@xsect{#5}}

\def\@begintheorem#1#2{\it \trivlist \item[\hskip \labelsep{\bf #1\ #2.}]}

\def\@opargbegintheorem#1#2#3{\it \trivlist
      \item[\hskip \labelsep{\bf #1\ #2\ (#3).}]}

\makeatother

\newif\ifshortconferences
\shortconferencesfalse
\newif\ifmediumconferences
\mediumconferencesfalse

\def\ending#1{{\count1=#1\relax
\count2=\count1
\divide\count2 by 100
\multiply\count2 by 100
\advance\count1 by -\count2
\ifnum\count1=11
th%
\else \ifnum\count1=12
th%
\else \ifnum\count1=13
th%
\else 
\count2=\count1
\divide\count1 by 10
\multiply\count1 by 10
\advance\count2 by -\count1
\ifnum\count2=1
st%
\else \ifnum\count2=2
nd%
\else \ifnum\count2=3
rd%
\else th%
\fi\fi\fi\fi\fi\fi
}}

\def\STOC{\conf{STOC}}

\def\Proceedings{\ifshortconferences Proc.\else\ifmediumconferences Proc.\else Proceedings\fi\fi}
\def\Proceedingsofthe{\ifshortconferences Proc.\else\ifmediumconferences Proc.\else Proceedings of the\fi\fi}

\newcounter{confnum}

\def\conf#1#2{%
\setcounter{confnum}{#2}%
\addtocounter{confnum}{-\csname #1zero\endcsname}%
\ifnum\value{confnum}=1%
\expandafter\ifx\csname #1One\endcsname\relax%
\Proceedingsofthe\ \arabic{confnum}\ending{\value{confnum}}\ \csname #1name\endcsname%
\else \csname #1One\endcsname\fi%
\else%
\Proceedingsofthe\
\arabic{confnum}\ending{\value{confnum}}\ \csname #1name\endcsname\fi}

\def\qsym{\vrule width0.7ex height0.9em depth0ex}

\newif\ifqed\qedtrue

\def\noqed{\global\qedfalse}

\def\qed{\ifqed{\penalty1000\unskip\nobreak\hfil\penalty50
\hskip2em\hbox{}\nobreak\hfil\qsym
\parfillskip=0pt \finalhyphendemerits=0\par\medskip}\fi\global\qedtrue}

\makeatletter
\def\eqnqed{\noqed
	\def\@tempa{equation}
	\ifx\@tempa\@currenvir\def\@eqnnum{\qsym}%
	\addtocounter{equation}{-1}\else%
    \def\@@eqncr{\let\@tempa\relax
    \ifcase\@eqcnt \def\@tempa{& & &}\or \def\@tempa{& &}%
      \else \def\@tempa{&}\fi
     \@tempa {\def\@eqnnum{{\qsym}}\@eqnnum}
     \global\@eqnswtrue\global\@eqcnt\z@\cr}\fi}

\def\eqnlabel#1#2{\if@filesw {\let\thepage\relax%
   \def\protect{\noexpand\noexpand\noexpand}%
   \edef\@tempa{\write\@auxout{\string
      \newlabel{#2}{{{#1}}{\thepage}}}}%
   \expandafter}\@tempa%
   \if@nobreak \ifvmode\nobreak\fi\fi\fi%
	\def\@tempa{equation}
	\ifx\@tempa\@currenvir\def\theequation{{#1}}%
	\addtocounter{equation}{-1}\else%
    \def\@@eqncr{\let\@tempa\relax
    \ifcase\@eqcnt \def\@tempa{& & &}\or \def\@tempa{& &}%
      \else \def\@tempa{&}\fi
     \@tempa {\def\@eqnnum{{#1}}\@eqnnum}
     \global\@eqnswtrue\global\@eqcnt\z@\cr}\fi}

\makeatother

\def\QED{\qed}

\makeatother

\squeeze

\newcommand{\PPOLY}{{\rm P/poly}}
\newcommand{\NT}{{\rm NTIME}}

\newcommand{\PIT}{{\rm PIT}}
\newcommand{\threeSAT}{{\rm 3SAT}}
\newcommand{\subNEXP}{{\rm NSUBEXP}}
\newcommand{\perm}{{\rm perm}}
\newcommand{\NCone}{{\rm NC1/poly}}
\begin{document}



\title{
Derandomizing Polynomial Identity 
over Finite Fields Implies Super-Polynomial Circuit Lower Bounds for
$\NEXP$
}

\author{Bin Fu
\\
Department of Computer Science\\
University of Texas--Pan American\\
 Edinburg, TX 78539, USA\\
bfu@utpa.edu\\\\
} \maketitle

\begin{abstract}
We show that
derandomizing polynomial identity testing over an arbitrary finite
field implies that $\NEXP$ does not  have  polynomial size boolean
circuits. In other words, for any finite field $F(q)$ of size $q$,
$\PIT_q\in \subNEXP\Rightarrow \NEXP\not\subseteq \PPOLY$, where
$\PIT_q$ is the polynomial identity testing problem over $F(q)$, and
$\subNEXP$ is the nondeterministic subexpoential time class of
languages. Our result is in contract to Kabanets and Impagliazzo's
existing theorem that derandomizing the polynomial identity testing
in the integer ring $Z$ implies that $\NEXP$ does have polynomial
size boolean circuits or permanent over $Z$ does not have polynomial
size arithmetic circuits.
\end{abstract}

\section{Introduction}

 The polynomial identity testing problem (PIT) is to test whether a polynomial computed by an
arithmetic circuit  is identical to zero. PIT problem plays a
significant role in the field of computational complexity. It is
known that every polynomial computed by a polynomial size circuit
can be determined if it is identical to zero by a polynomial time
randomized algorithm~\cite{Schwartz80,Zippel79}.

The results of Impagliazzo and
Widgerson~\cite{ImpagliazzoWigderson97} suggested that every
randomized polynomial time algorithm can be derandomized into a
deterministic polynomial time algorithm. They proved that $\P=\BPP$
if E contains any problem that requires $2^{\Omega(n)}$ size boolean
circuits. It has been a long standing open problem in complexity
theory to separate NEXP from BPP.
Kabanets, Impagliazzo and Wigderson showed that derandomizing
Promise-BPP implies $\NEXP\not\subseteq
\PPOLY$~\cite{ImpagliazzoKabanetsWigderson}. Building upon the work
of Kabanets, Impagliazzo and
Wigderson~\cite{ImpagliazzoKabanetsWigderson}, Kabanets and
Impagliazzo~\cite{KabanetsImpagliazzo04} proved that to derandomize
the polynomial identity testing problem in the integer ring, one
must prove that $\NEXP$  has no polynomial size boolean circuits or
permanent has no polynomial size arithmetic circuits.
The proof of Kabanets and Impagliazzo's theorem was
simplified by Aaronson and Melkebeek~\cite{AaronsonMelkebeek10}.
Many papers have been
published toward the derandomization of the polynomial identity
problems~(see for examples,
\cite{AgrawalBiswas03,KlivansSpielman01,KabanetsImpagliazzo04,LewinVadhan98,ChenKao00,ShpilkaVolkovich08,LiptonVishnoi03,Zippel79,Schwartz80,RazShpilk05,KayalSaxena07,KarninShpilka07,DvirShpilk07a,Saxena08}).

We have not found any existing result that shows derandomization of
PIT over a finite field implies $\NEXP\not\subseteq \PPOLY$. It is
essential to identify the connection between derandomizing PIT over
finite fields and complexity classes separation. In this paper, we
study the implication of polynomial identity problem over finite
fields to separations in computational complexity theory. Our
results are derived without using permanent problem, and give a
direct implication for complexity classes separation via
derandomization of PIT.

We show that for any finite field $F$, if $\PIT$ over $F$ is in
$\subNEXP$, then $\NEXP\not\subseteq \PPOLY$. We also show that if
$\PIT$ over a finite field $F$ is in $\NP$, then $\NT(n^{\log
n})\not\subseteq \PH\cap \PPOLY$. It implies that if there exists a
polynomial time deterministic algorithm for polynomial identity
testing problem over any finite field $F$, then $\NT(n^{\log
n})\not\subseteq \BPP$.

Our proof is different from Kabanets and
Impagliazzo's~\cite{KabanetsImpagliazzo04} work that is for PIT over
$Z$.
Their methods involve
permanent that is $\#\P$-hard~\cite{Valiant},
and Toda's theorem $\PH\subseteq \P^{\#\P}$~\cite{TodaPPPH}. Our
method works for the PIT over any finite fields, but it does not
imply that derandomizing PIT over $Z$ separates NEXP from $\PPOLY$.
Therefore, Kabanets and Impagliazzo's~\cite{KabanetsImpagliazzo04}
work and this paper are complementary to each other to support the
importance of derandomizing the polynomial identity testing problem,
and its implication to nonuniform lower bounds.

\section{Notations}

A {\it boolean circuit} is a circuit with AND ($\bigwedge$), OR
($\bigvee$),  and NEGATION ($\bar{x}$) gates with fan-in at most
two, and no feedback. An {\it arithmetic circuit} is a circuit with
$+,-$ and $*$ gates over a finite field. The {\it size} of a circuit
$C(.)$ is the number of gates and is denoted by $|C(.)|$.

The {\it polynomial identity testing problem} is to test if a
polynomial computed by an arithmetic circuit is identical to zero.
We use $\PIT_q$ to represent the polynomial identity testing problem
over the field $F(q)$ of size $q$. Let $\PIT_Z$ be the polynomial
identity testing problem over the integers $Z$, which is an integral
ring.

The basic knowledge of algebra can be found in standard algebra
textbooks such as~\cite{Hungerford74}. Every finite field $F(q)$ of
size $q$ has $q=p^k$ for some prime number $p$ and integer $k$. For
an element $a$ in a finite field $F(q)$, its {\it order} is the
least integer $r\ge 1$ with $a^r=1$.

The permanent maps square matrices to values. Let
$A=(a_{i,j})_{n\times n}$ be an $n\times n$ matrix over integers.
Define permanent to be the function
$\perm(A)=\sum_{\sigma}\prod_{i=1}^n a_{i,\sigma(i)}$, where
$\sigma$ is over all permutations of $1,2,\cdots, n$.

Define $N=\{0,1,2,\cdots\}$ to be the set of nonnegative integers.
Let $t(n):N\rightarrow N$ be a nondecreasing function. Define
 $\NT(t(n))$ to be the class of languages accepted by nondeterministic
Turing machines in time $O(t(n))$.
For a function $f(n): N\rightarrow N$, it is {\it time
constructible} if given an integer $n$, $f(n)$ can be computed in
$O(f(n))$ steps by a deterministic Turing machine.
Theorem~\ref{Zak-theorem} is a separation of nondeterministic
complexity classes due to Zak\cite{Zak83}.

\begin{theorem}[\cite{Zak83}]\label{Zak-theorem}
If $t_1(.)$ and $t_2(.)$ are time-constructible nondecreasing
functions from $N$ to $N$, and $t_1(n+1)=o(t_2(n))$, then
$\NT(t_2(n))$ is strictly contained in $\NT(t_1(n))$.
\end{theorem}

Assume that $M(.)$ is an oracle Turing machine.  A decision
computation
 $M^A(x)$ returns either $0$ or $1$ when the input is $x$ and
oracle is $A$. For a class $C$ of languages, we use $\NP^C=\NP_{\rm
T}(C)$ to represent the class of languages that can be reducible to
the languages in $C$ via polynomial time nondeterministic Turing
reductions. Define $\NEXP=\cup_{c=1}^{\infty} \NT(2^{n^c})$ and
$\NP=\cup_{c=1}^{\infty} \NT(n^c)$.

Let $\PH$ be the class of polynomial time
hierarchy~\cite{Stockmeyer} $\PH=\cup_{i=1}^{\infty} \sum_i^{\rm
P}$, where $\sum_1^{\rm P}=\NP$, and $\sum_{i+1}^{\rm
P}=\NP^{\sum_i^{\rm P}}$ for all $i\ge 1$. Define the subexponential
time nondeterministic class to be
$\subNEXP=\cap_{\epsilon>0}\NT(2^{n^{\epsilon}})$. Define $\PPOLY$
to be the class of languages that have nonuniform polynomial size
circuits. BPP, which stands for bounded-error probabilistic
polynomial time, is the class of decision problems solvable by a
probabilistic Turing machine in polynomial time, with an error
probability of at most 1/3 for all instances.

An {\it instance of \threeSAT} is a 3CNF that is a conjunction of
clauses of at most three literals. For example, $(x_1\bigvee
\overline{x_2}\bigvee x_3)\bigwedge (\overline{x_1}\bigvee
\overline{x_2}\bigvee x_3) \bigwedge (x_1\bigvee
\overline{x_4}\bigvee x_5)$. A formula is said to be satisfiable if
it can be made true by assigning appropriate logical values (i.e.
TRUE (1), FALSE(0)) to its variables. The \threeSAT~is, given a
3CNF, to check whether it is satisfiable. It is well known that
\threeSAT~ is NP-complete problem~\cite{Cook-NP-complete}. The
number of variables of a \threeSAT~instance and its length is
polynomially related.

\section{Our Results}

Theorem~\ref{main-thm} is the main theorem of this paper. It will be
proved in Section~\ref{proof-section}. Theorem~\ref{main-thm} is
stated in a format so that we have a self-contained proof. Some
corollaries that involve some existing results are stronger than the
main theorem.

\begin{theorem}\label{main-thm}
Let $t(n)$ and $t'(n)$ be time constructible nondecreasing
superpolynomial functions from $N$ to $N$ with $t'(n+1)=o(t(n))$.
Let $h(n)$ be an nondecreasing function from $N$ to $N$ such that
for every fixed $c>0$, $h(n^{c})+n^c\le t'(n)$ for all large $n$.
Let $F(q)$ be a finite field of size $q$. If $\PIT_q\in \NT(h(n))$,
then $\NT(t(n))\not\subseteq \NP^{\NP}\cap \PPOLY$.
\end{theorem}

Using the main theorem, we have some corollaries. Their proofs need
to combine Theorem~\ref{main-thm} with some existing well known
theorems in the computational complexity theory.

\begin{corollary}\label{coro-first-corollary}
Let $t(n)$ and $t'(n)$ be time constructible nondecreasing
superpolynomial functions from $N$ to $N$ with $t'(n+1)=o(t(n))$.
Let $h(n)$ be a nondecreasing function from $N$ to $N$ such that for
every fixed $c>0$, $h(n^{c})+n^c\le t'(n)$ for all large $n$. Let
$F(q)$ be a finite field of size $q$. If $\PIT_q\in \NT(h(n))$, then
$\NT(t(n))\not\subseteq \PH\cap \PPOLY$.
\end{corollary}


\begin{proof}
Assume $\NT(t(n))\subseteq \PH\cap \PPOLY$. By Karp and Lipton's
thereom~\cite{KarpLipton}, we have $\PH=\sum_2^P=\NP^{\NP}$. It
follows from Corollary~\ref{coro-first-corollary}.
\end{proof}

\begin{corollary}
If $\PIT_q\in \subNEXP$ for a finite field $F(q)$, then
$\NEXP\not\subseteq \PPOLY$.
\end{corollary}

\begin{proof}
Assume that $\PIT_q\in \subNEXP$ and $\NEXP\subseteq \PPOLY$. By
Impagliazzo, Kabanets, and Wigderson's
theorem~\cite{ImpagliazzoKabanetsWigderson}, $\NEXP=\PH$. We have a
contradiction by Theorem~\ref{main-thm}.
\end{proof}

\begin{corollary}\label{NEXP-BPP-corollary}
If $\PIT_q\in \subNEXP$  for a finite field $F(q)$, then $\NEXP\not=
\BPP$.
\end{corollary}

\begin{proof}
Assume that $\PIT_q\in \subNEXP$ and $\NEXP=\BPP$. It is well known
that Adleman~\cite{Adleman} proved
 $\BPP\subseteq \NP^{\NP}$. We have a contradiction by Theorem~\ref{main-thm}.
\end{proof}

\begin{corollary}
If $\PIT_q\in \NP$  for a finite field $F(q)$, then $\NT(n^{\log
n})\not\subseteq \BPP$.
\end{corollary}

\begin{proof}
It is similar to the proof of Corollary~\ref{NEXP-BPP-corollary}.
\end{proof}

\section{Overview of Our Method}

In this section, we give a brief review of our method.  The main
theorem  will be proved by contradiction.
A special version of our main
theorem is formulated as $\PIT_2\in \NP\Rightarrow
\NEXP\not\subseteq \NP^{\NP}\cap \PPOLY$.


Assume $\PIT_2\in \NP$ and $\NEXP\subseteq \NP^{\NP}\cap \PPOLY$.
Let $K$ be a complete language of the class $\NEXP$ and let
\threeSAT~be computed by a polynomial size circuit $C(.)$.  Our main
technical contribution is a method that transforms a boolean circuit
$C(.)$ into an arithmetic circuit $A^*_C(.)$ over a finite field
$F(q)$ such that $C(.)$ decides \threeSAT~if and only if $A^*_C(.)$
is identical to zero.

As the \threeSAT~problem is not arithmetically defined as permanent.
If each instance of \threeSAT~is encoded as a binary string that
will be easy to decode, then there are some binary strings that do
not encode valid instances of \threeSAT. In other words, a mapping
from instances of \threeSAT~to binary strings may not be both
one-one and onto.
 We construct a
 special polynomial size arithmetic function $G(Y)$ that is zero  if $Y$ is not an instance of \threeSAT,
 and nonzero otherwise. For an instance $f(x_1,x_2,\cdots,x_n)$ of \threeSAT,
 it is satisfiable if and only if  at least one of $f(0,x_2,\cdots,x_n)$ and
 $f(1,x_2,\cdots,x_n)$ is satisfiable. This recursive relation is
 also converted into an arithmetic circuit $A^*_C(.)$ as PIT problem to verify whether  $C(.)$ decides \threeSAT. The arithmetic circuit $A^*_C(.)$ is expressed as $H(f)G(f)$.  The arithmetic circuit $H(f)$ is used to verify
the recursive relationship of
 circuit $C(.)$ for deciding \threeSAT. As the input of $A^*_C(.)$ has many cases that do not encode any instance of \threeSAT, the function $G(.)$ has a value zero to pass the identity testing among those cases.

We will show that $K$ can be computed by $M^{\threeSAT}(.)$ for a
polynomial time oracle Turing machine $M(.)$. A polynomial time
nondeterministic computation will be derived to compute $K$.  A
circuit $C(.)$ will be guessed and is checked via converting to
$\PIT_2$, which is verified again in a nondeterministic polynomial
time. Thus, we have $K\in \NP$. This contradicts the well know
nondeterministic computational complexity hierarchy, which is stated
in Theorem~\ref{Zak-theorem} and implies $\NEXP\not= \NP$.

This approach lets us obtain lower bound for NEXP under the
existence  of the derandmoization of PIT over an arbitrary finite
field without using permanent that is $\#\P$ hard. This paper has
almost self-contained proof for the main theorem. A reader is able
to understand our main theorem just by knowing
Theorem~\ref{Zak-theorem} and that \threeSAT~is
NP-complete~\cite{Cook-NP-complete}, which can be found in a
standard textbook of theory of computation.

\section{Proof of Main Theorem}\label{proof-section}

In this section, we derive our main theorem. Some lemmas are
provided to convert boolean circuits into arithmetic circuits. It is
divided into several subsections to prove the main theorem.

\subsection{From Boolean Circuits to Arithmetic Circuits}

In this section, we show how to transform  boolean circuits into an
arithmetic circuits.

For a finite field $F(q)$, we have the following property that is
often called ``Fermat Little Theorem" for the case that $q$ is a
prime number. Its proof can be found in a standard algebra textbook.
For completeness, its proof is included here.

\begin{lemma}\label{classical-lemma}
Let $F(q)$ be a finite field. For any $a\in F(q)-\{0\}$,
$a^{q-1}=1$.
\end{lemma}

\begin{proof} Assume that $F$ is a finite field.
Let $[a]=\{a, a^2, a^3,\cdots, \}$ be the set of elements in $F(q)$
generated by $a$. $([a],.)$ forms a subgroup of $F(q)^*=F(q)-\{0\}$,
where ``." is the multiplication operation over field $F(q)$.
Therefore, the order $r$ of $a$ is the size of $[a]$. Therefore, $r$
is a divisor of $q-1$. So, $a^r=a^{q-1}=1$.

\end{proof}

We give Lemma~\ref{encoding-lemma} to convert an instance
of~\threeSAT~into a binary string. We give
Definition~\ref{normalized-def} to normalize the input of an
instance of~\threeSAT.

\begin{definition}\label{normalized-def}
 Assume that an instance $C_1\bigwedge C_2\bigwedge\cdots \bigwedge C_m$ of \threeSAT~ of $n$ variables and satisfies the conditions
 below:
\begin{enumerate}[1.]
\item
each $C_i$ has at most three literals,
\item
no variable appears in two literals of the same clause,
\item
 all
clauses have a different set of literals, and
\item
its $n$ variables are $x_1,x_2,\cdots,x_n$ that are indexed from $1$
to $n$
\end{enumerate}
We have the following definitions:
\begin{itemize}
\item
 Define $E_l(x_i)=(i, 1)$ and
$E_l(\overline{x_i})=(i, 0)$.
\item
Define $E_c((y_i\bigvee y_j\bigvee
   y_k)))=(E_l(y_i), E_l(y_j), E_l(y_k))$ for each
clause
   $(y_i\bigvee y_j\bigvee y_k)$.
\item
Define the normalized representation of $C_1\bigwedge
C_2\bigwedge\cdots \bigwedge C_m$  of \threeSAT~to be $(E_c(C_1),
E_c(C_2),\cdots, E_c(C_m))$.
\item
The logical value TRUE ($1$) is treated as special instance of 3SAT,
and we define its normalized representation to be $(1)$.
\item
The logical value FALSE ($0$) is treated as special instance of
3SAT, and we define its normalized representation to be $(0)$.
\end{itemize}
\end{definition}

\begin{lemma}\label{encoding-lemma}
Assume an instance $f$ of \threeSAT~ is a normalized representation
as Definition~\ref{normalized-def}.
\begin{enumerate}
\item\label{case1}
There is a polynomial time encoding method $E(.)$ such that given an
instance $f$ of at most $n$ variables of \threeSAT, $E(n,f)$ is a
0,1-string of length $8n^4$.
\item\label{case2}
There is a polynomial time decoding method $D(.)$ such that given a
0,1-string $s=E(n,f)$ for some instance $f$ with at most $n$
variables of \threeSAT, $D(s)=f$.
\item\label{case3}
There is a polynomial time algorithm $H(.)$ such that $H(1^n)$
generates a polynomial size boolean circuit $V_n(.)$ such that given
a
 0,1-string $s$ of length $8n^4$, $V_n(s)= 1$ if $s=E(n,f)$ for some instance of at most  $n$ variables of \threeSAT, and 0 otherwise.
\end{enumerate}
\end{lemma}

\begin{proof} We prove the three statements below:

Statement~\ref{case1}: Given a normalized representation an instance
of \threeSAT, just replace each symbol with ASCI table to transfer
it into a binary string. Append $10^k$ for some $k$ so that the
total length is exactly equal to $8n^4$. Each 3CNF instance has at
most $24{n\choose 3}<4n^3$ different clauses. $8n^4$ binary bits are
enough to encode any 3CNF instance of at most $n$ variables.

Statement~\ref{case2}: It is straight forward to decode the binary
string into an instance of \threeSAT~ by using the ASCI table.

Statement~\ref{case3}: With a polynomial time, we can check if a
binary string is a binary string to encode an valid instance of a
\threeSAT. It can be converted into a polynomial size boolean
circuit.

\end{proof}

\begin{definition}
Let $C(x_1,x_2,\cdots, x_n):\{0,1\}^n\rightarrow \{0,1\}$ be a
boolean circuit, and $A(y_1,y_2,\cdots, y_n): F(q)^n\rightarrow
F(q)$ be an arithmetic circuit over a finite field $F(q)$. We say
$C(.)$ and $A(.)$ are {\it equivalent } if for any $a_1,a_2,\cdots,
a_n\in \{0,1\}$, $C(a_1,a_2,\cdots, a_n)=0 \Leftrightarrow
A(a_1,a_2,\cdots, a_n)=0$ in the field $F(q)$; and
$C(a_1,a_2,\cdots, a_n)=1 \Leftrightarrow A(a_1,a_2,\cdots, a_n)=1$
in the field $F(q)$.
\end{definition}

The following Lemma~\ref{bool-arithmetic-lemma} shows how a boolean
circuit is converted into an equivalent arithmetic circuit with a
similar size.

\begin{lemma}\label{bool-arithmetic-lemma}
For any boolean circuit $C(.)$, then there is an equivalent
arithmetic circuit $A_C(.)$ over a field $F(q)$ such that
$|A_C(.)|=O(|C(.)|)$. Furthermore, $A_C(.)$ can be constructed from
$C(.)$ in a polynomial time of $|C(.)|$.
\end{lemma}

\begin{proof} We just show how to simulate the three AND, OR, and
NOT gates in a boolean circuit with arithmetic operations. The
arithmetic circuit is constructed by simulating the boolean circuit
$C(.)$ gate by gate. For an AND operation $a\bigwedge b$, it can be
converted into product $a\cdot b$ over $F(q)$. For an OR operation
$a \bigvee b$, it can be converted into $1-(1-a)(1-b)$. For an NOT
operation $\neg{a}$, it is converted into $1-a$. Since each gate in
$C(.)$ is transformed into
 $O(1)$ gates in $A_C(.)$, we have $|A_C(.)|=O(|C(.)|)$. It is easy to
see that the total time to construct $A_C(.)$ is a polynomial time
of $|C(.)|$.
\end{proof}

\begin{definition}\label{nor-encoding-def}
Let $f$ be a normalized representation of an instance of \threeSAT.
Define $E(f)$ to be the {\it normalized binary encoding} of $f$,
where $E(.)$ is as defined in Lemma~\ref{encoding-lemma}. Define
$one_n=E(n, (1))$ and $zero_n=E(n,(0))$ for the normalized binary
representation of true and false respectively, where $E(.)$ is given
in Lemma~\ref{encoding-lemma}.
\end{definition}

\begin{lemma}\label{exclude-lemma}
Let $F(q)$ be a fixed finite field. Then there is a polynomial time
algorithm that given an unary integer $1^n$, it generates an
arithmetic circuit $G_n(x_1,x_2,\cdots,x_m)$ with $m=8n^4$ such that
\begin{enumerate}
\item
$G_n(x_1,x_2,\cdots,x_m)=0$ if at least one of $x_1,x_2,\cdots, x_m$
is not in $\{0,1\}$;
\item
$G_n(x_1,x_2,\cdots,x_m)=0$ if $x_1x_2\cdots x_m$ is not a
normalized binary encoding of an instance of \threeSAT~ with at most
$n$ variables; and
\item
$G_n(x_1,x_2,\cdots,x_m)\not=0$ if $x_1x_2\cdots x_m$ is a
normalized binary encoding of an instance of \threeSAT~ with at most
$n$ variables.
\end{enumerate}
\end{lemma}

\begin{proof} By Lemma~\ref{encoding-lemma}, we let $V_n(f)$ be a boolean circuit such that $V_n(f)\not=0$ if and only if
$f$ is a normalized binary encoding of an instance of~\threeSAT. Let
$A_V(.)$ be the arithmetic circuit defined by
Lemma~\ref{bool-arithmetic-lemma}.

Define $R(x)=1-(x(x-1))^{q-1}$.
It is easy to see that $R(x)\not=0$ if and only if $x\in \{0,1\}$ by
Lemma~\ref{classical-lemma}.

Finally, we define $G_n(x_1,x_2,\cdots,x_m)=R(x_1)R(x_2)\cdots
R(x_m)A_{V_n}(x_1,x_2,\cdots, x_m)$. It is easy to see that $G_n(.)$
satisfies expected properties.
\end{proof}

\begin{lemma}\label{reducing-lemma}
Assume that each input instance of \threeSAT~ is a normalized binary
encoding (see Definition~\ref{nor-encoding-def}. Then there is a
polynomial time algorithm such that given $1^n$, it generates
$n^{O(1)}$ size arithmetic circuits $S_{n,0}(.)$, and $S_{n,1}(.)$
such that the following are satisfied:
\begin{enumerate}
\item
$S_{n,0}(f)$ generates a normalized binary encoding for $g(0,
x_2,\cdots, x_k)$ if $f$ is a normalized binary encoding of a
\threeSAT~ instance $g(x_1, x_2,\cdots, x_k)$ with $0\le k\le n$;
\item
 $S_{n,1}(f)$ generates a
normalized binary encoding for $g(1, x_2,\cdots, x_k)$ if $f$ is a
normalized binary encoding of a \threeSAT~ instance $g(x_1,
x_2,\cdots, x_k)$ with $0\le k\le n$; and

\item
$S_{n,i}(f)=f$ for $i\in \{0,1\}$ and $f\in\{zero_k,one_k\}$ with
$0\le k\le n$.
\end{enumerate}
\end{lemma}

\begin{proof}
It is easy to see that there is a polynomial time algorithm to
generate the formulas $g(0, x_2,\cdots, x_k)$, $g(1, x_2,\cdots,
x_k)$ with $0\le k\le n$.
Thus, we can get a boolean circuits to generate
them. By Lemma~\ref{bool-arithmetic-lemma}, we can get the
equivalent arithmetic circuits to generate them respectively.
\end{proof}

\subsection{From Arithmetic Circuits to PIT}

In this section, we show how to convert the arithmetic expressions
developed in the last section and a circuit for \threeSAT~into a PIT
problem.

The following Lemma~\ref{trans-arithm-lemma} shows how to use the
PIT problem over a finite field to check if a boolean circuit
decides \threeSAT. It transform a boolean circuit into an arithmetic
circuit in a polynomial number of steps.

\begin{lemma}\label{trans-arithm-lemma}Let $F(q)$ be a fixed finite
field. Then there is a polynomail time algorithm such that given a
circuit $C_n(.)$, it generates another arithmetic circuit
$A^*_{C_n}(.)$ over a finite field $F(q)$  such that $C_n(.)$
decides  instances for \threeSAT~with at most $n$ variables if and
only if $A^*_{C_n}(.)$ is identical to zero.
\end{lemma}

\begin{proof} We assume that all instances of \threeSAT~ with at most $n$
variables have normalized binary encoding of length $8n^4$ as input
for $C_n(.)$. Let $S_{n,0}(.)$ and $S_{n,1}(.)$ be defined as in
Lemma~\ref{reducing-lemma}. Let $A_{C_n}(.)$ be the arithmetic
circuit that is equivalent to  $C_n(f)$ by
Lemma~\ref{bool-arithmetic-lemma}. Let $one_n$ be the normalized
binary encoding of logical constant TRUE $(1)$, and let $zero_n$ be
the normalized binary encoding of logical constant FALSE $(0)$ (see
Definition~\ref{nor-encoding-def}). Let $y_0, y_1, y_2$
be new variables that do not appear in $A_{C_n(.)}$.  We have the
arithmetic circuit
\begin{eqnarray*}
H(f, y_0, y_1, y_2
)=&&
y_0(A_{C_n}(one_n)-1)+y_1A_{C_n}(zero_n)+\\
&&y_2(A_{C_n}(f)-(1-(1-A_{C_n}(S_{n,0}(f)))(1-A_{C_n}(S_{n,1}(f)))).
\end{eqnarray*}

Let $G_n(.)$ be the arithmetic circuit defined by
Lemma~\ref{exclude-lemma}. Define $A^*_{C_n}(f, y_0, y_1, y_2)=H(f,
y_0, y_1, y_2)G_n(f)$.

Assume that circuit $C_n(.)$ decides \threeSAT~for all instance of
at most $n$ variables, and takes the normalized binary encoding of
length $8n^4$ as input. For each normalized binary encoding $f$ of
an instance of \threeSAT, we have $H(f, y_0, y_1, y_2)=0$. This is
because recursive relation for each decider of \threeSAT. If $f$ is
not a normalized binary encoding of a valid instance of \threeSAT,
we have $G_n(f)=0$. Therefore, $A^*_{C_n}$ is identical to zero.

Assume that $A^*_{C_n}$ is identical to zero. We need to verify that
${C_n}(.)$ is a circuit for \threeSAT. For each valid instance $f$
with at most $n$ variables of \threeSAT, we have $G_n(f)\not=0$ by
Lemma~\ref{exclude-lemma}. Thus, $H(f, y_0, y_1, y_2)=0$. It
confirms the $C_n(.)$ satisfies the recursive relation for a
\threeSAT~decider. For each instance $g(x_1,\cdots, x_n)$ with $n$
variables for \threeSAT, let $a_1,\cdots, a_k\in \{0,1\}$ be an
assignment for its first $k$ variables with $0\le k\le n$. We can
still find a normalized binary encoding $f_k$ for $g(a_1,\cdots,a_k,
x_{k+1},\cdots, x_n)$ and $f_k$ has length $8n^4$. Since
$G_n(f_k)\not=0$ and $H(f_k, y_0, y_1, y_2)=0$, $C_n(.)$ satisfies
the recursive relation for a \threeSAT~decider at the cases
$g(a_1,\cdots,a_k, x_{k+1},\cdots, x_n)$ for all $0\le k\le n$..


We have that ${C_n}(.)$ is a circuit for \threeSAT~if and only if
$A^*_{C_n}(f, y_0, y_1, y_2)$ is zero since it verifies if the
circuit ${C_n}(.)$ satisfies the recursion for~\threeSAT~instance
satisfiability.
\end{proof}

\subsection{From Derandomization to Separations}
In this section, we show that derandomizing PIT over a finite field
implies separation of computational complexity classes. The proof of
main theorem is given here.

\begin{proof}[Theorem~\ref{main-thm}]
Assume $\NT(t(n))\subseteq \NP^{\NP}\cap \PPOLY$.
Let $K$ be an arbitrary language of $\NT(t(n))$ under polynomial
time many-one reduction. Let $M^{\threeSAT}(.)$ be a polynomial time
nondeterministic Turing machine to accept $K$, and runs in a
polynomial time $p(n)$ that is nondecreasing function from $N$ to
$N$. Let $N(.)$ be an $O(h(n))$ time nondeterministic Turing machine
that decides $\PIT_q$. We have the following nondeterministic
algorithm for $K$.

\vskip 10pt

{\bf Nondeterministic Algorithm for $K$}

Input $x$ of length $n$,

\begin{enumerate}[1.]
\item
\qquad Guess a circuit $C_{p(n)}(.)$  for deciding the instances of
variables at most $p(n)$ for \threeSAT;

\item
\qquad Generate an arithmetic circuit $A^*_{C_{p(n)}}(.)$ $\PIT_q$
problem to verify $C_{p(n)}(.)$  by Lemma~\ref{trans-arithm-lemma};

\item\label{verify-step}
\qquad Run $N(A^*_{C_{p(n)}}(.))$ nondeterministically to decide if
$A^*_{C_{p(n)}}(.)$ is identical to zero;

\item
\qquad Nondeterministically select a path $P^*$ in
$M^{\threeSAT}(x)$;

\item
\qquad If step~\ref{verify-step} is successful, use
$A^*_{C_{p(n)}}(.)$ to answer all the queries to~\threeSAT~in path
$P^*$;

\item
\qquad Output yes, if $P^*$ accepts;

\end{enumerate}

{\bf End of Algorithm}

\vskip 10pt

Since $M(.)$ runs in polynomial time $p(n)$, the instance of queries
made by $M(.)$ has at most $p(n)$ variables. Since $\NT(t(n))\in
\PPOLY$ implies \threeSAT$\in \PPOLY$, there is a polynomial size
boolean circuit $C_{p(n)}(.)$ to decide \threeSAT~ for all instances
with at most $p(n)$ variables. Let $q_1(n)$
 be a polynomial with $|C_{p(n)}(.)|\le q_1(n)$. We have an
  arithmetic circuits $A^*_{C_{p(n)}}(.)$ that is equivalent with
 $C_{p(n)}(.)$ by Lemma~\ref{bool-arithmetic-lemma}. We also have $|A^*_{C_{p(n)}}(.)|\le
 q_2(n)$ for some polynomial $q_2(n)$. Step~\ref{verify-step} in the
 algorithm takes $O(h(q_2(n)))$ nondeterministic steps. For some constant $c>0$, the entire
 computation is in $O(h(n^c)+n^c)=O(t'(n))$ nondeterministic steps.

The nondeterministic algorithm above shows that $K$ is in
$\NT(t'(n))$. Since $K$ is an arbitrary language in $\NT(t(n))$, we
have $\NT(t(n))\subseteq \NT(t'(n))$. This contradicts the well
known hierarchy theorem (see Theorem~\ref{Zak-theorem}) for
nondeterministic computation classes.

\end{proof}

\section{Generalization to Bounded Depth Circuits}

In this section, we consider the problem for the $\PIT$ with bounded
depth arithmetic circuits, and its connection to the
super-polynomial lower bounds of bounded depth boolean circuits. It
is an open problem to prove $\NEXP\not\subseteq \NCone$, where
$\NCone$ is the class of languages that have polynomial size $O(\log
n)$-bounded depth boolean circuits.

\begin{definition}
Let $d(n)$ be a function from $N$ to $N$. Let $F(q)$ be a finite
field of size $q$. Define $\PIT_q(d(n))$ to be the polynomial
identity testing problem that decides if a polynomial computed by an
arithmetic circuit of depth at most $d(n)$ is identical to zero over
field $F(q)$.
\end{definition}

\begin{definition}
Let $d(n)$ be a function from $N$ to $N$. A {\it $d(n)$-bounded
depth} boolean circuits is the class of boolean circuits that
consists of AND, OR, and NOT gates with unbounded fan-in for AND and
OR gates. Define {\it Depth$(d(n))$-PC} to be the class of languages
that have polynomial size $d(n)$-bounded depth boolean circuits.
\end{definition}

\begin{theorem}\label{main2-thm}
Let $t(n)$ and $t'(n)$ be time constructible nondecreasing
superpolynomial functions from $N$ to $N$ with $t'(n+1)=o(t(n))$.
Let $h(n)$ be a nondecreasing function from $N$ to $N$ such that for
every fixed $c>0$, $h(n^{c})+n^c\le t'(n)$ for all large $n$. Let
$F(q)$ be a finite field of size $q$. Let $d(n)$ be a function from
$N$ to $N$ with $d(n)\ge \log n$. If $\PIT_q(O(d(n))\in \NT(h(n))$,
then $\NT(t(n))\not\subseteq \NP^{\NP}\cap $ Depth$(d(n))$-PC.
\end{theorem}

\begin{proof}[Sketch]
The proof is similar to that of Theorem~\ref{main-thm}. We need to
have a similar lemma like Lemma~\ref{bool-arithmetic-lemma} to show
that a bounded  depth $k$ boolean circuit has an equivalent bounded
depth $O(k)$ arithmetic circuit. With $d(n)\ge \log n$, we also have
a lemma similar to Lemma~\ref{trans-arithm-lemma} such that a
bounded $d(n)$ depth boolean
 circuit for~\threeSAT~can be converted into a
$\PIT_q(O(d(n)))$ problem to verify it. This is because $S_{n,0}(.)$
and $S_{n,1}(.)$ have polynomial size $O(\log n)$-depth boolean
circuits.
\end{proof}

\begin{corollary}
Let $t(n)$ and $t'(n)$ be time constructible nondecreasing
superpolynomial functions from $N$ to $N$ with $t'(n+1)=o(t(n))$.
Let $h(n)$ be an nondecreasing function from $N$ to $N$ such that
for every fixed $c>0$, $h(n^{c})+n^c\le t'(n)$ for all large $n$.
Let $d(n)$ be a function from $N$ to $N$ with $d(n)\ge \log n$. Let
$F(q)$ be a finite field of size $q$. If $\PIT_q(O(d(n)))\in
\NT(h(n))$, then $\NT(t(n))\not\subseteq \PH\cap $
Depth$(O(d(n)))$-PC.
\end{corollary}

\begin{proof}
Assume $\NT(t(n))\subseteq \PH\cap$Depth$(O(d(n)))$-PC. By Karp and
Lipton's thereom~\cite{KarpLipton}, we have
$\PH=\sum_2^P=\NP^{\NP}$. It follows from Theorem~\ref{main2-thm}.
\end{proof}

\begin{corollary}
If $\PIT_q(O(\log n))\in \subNEXP$ for a finite field $F(q)$, then
$\NEXP\not\subseteq \NCone$.
\end{corollary}

\begin{proof}
Assume that $\PIT_q\in \subNEXP$ and $\NEXP\subseteq
\NCone=$Depth$(O(\log n))$-PC. By Impagliazzo, Kabanets, and
Wigderson's theorem~\cite{ImpagliazzoKabanetsWigderson},
$\NEXP=\PH$. We have a contradiction by Theorem~\ref{main2-thm}.
\end{proof}

\section{Conclusions
}

The result developed in this shows that derandomizing the PIT in any
finite field implies NEXP does not have nonuniform polynomial size
circuits. It gives right motivation to study the derandomization of
PIT in finite fields that the computational complexity community has
spent much efforts. We hope that the results in this paper brings a
tool to achieve the separation of NEXP from BPP via derandomizing
$\PIT_p$ for a prime number $p$ such as $2$. Since there exists an
oracle to collapse NEXP to BPP by Heller~\cite{Heller}, separating
$\NEXP$ from $\BPP$ requires a new way to go through the barrier of
relativization.

Another interesting open problem is if derandomizing $\PIT$ over $Z$
implies $\NEXP\not\subseteq\PPOLY$ (In other words, $\PIT_Z\in
\subNEXP\Rightarrow \NEXP\not\subseteq\PPOLY$?). Our technology
fails on integers $Z$. We cannot obtain a similar result as
Lemma~\ref{exclude-lemma} over the ring $Z$ of integers.


\vskip 10pt

{\bf Acknowledegments:} The author is grateful to Bohan Fan, Cynthia
Fu, and Feng Li for their proofreading and suggestions for an
earlier version of this paper. This research is supported in part by
NSF Early Career Award CCF-0845376.


\end{document}